\documentclass[a4paper,10pt]{amsart}

\usepackage{amssymb}
\usepackage{enumerate}
\usepackage{amsthm}
\usepackage{amsmath}
\usepackage{algorithm,algorithmic}

\newtheorem{theorem}{Theorem}
\newtheorem{lemma}{Lemma}
\theoremstyle{remark}

\newcommand{\nat}{\mathbb{N}}

\begin{document}

\title{{\bf Sums of powers via integration}}

\author{M. Torabi Dashti}
\address{M. Torabi Dashti, Dept. Computer Science, ETH Z{\"u}rich, Switzerland.}

\begin{abstract}
Sum of powers $1^p+\cdots+n^p$, with $n,p\in\nat$ and $n\ge 1$, can be
expressed as a polynomial function of $n$ of degree $p+1$.  Such
representations are often called Faulhaber formulae. A simple
recursive algorithm for computing coefficients of Faulhaber formulae
is presented. The correctness of the algorithm is proved by giving a
recurrence relation on Faulhaber formulae.

\bigskip

\noindent
\emph{Keywords:} Faulhaber formulae; recurrence relation
\end{abstract}

\maketitle

\section{Introduction}
\label{sec:intro}
Define $f_p(n)=1^p+\cdots+n^p$, for $p\in\nat,n\in\nat^+$. One can
express $f_p(n)$ as a polynomial function of $n$ of degree $p+1$. For
example,
$1^2+\cdots+n^2=\frac{1}{3}n^3+\frac{1}{2}n^2+\frac{1}{6}n$. Such
representations are often called Faulhaber formulae, after Johann
Faulhaber (1580--1635). In this paper, we study the following problem:

\vspace{0.5ex}
\begin{quote}
For $p\in\nat$, find the Faulhaber
formula that represents $1^p+\cdots+n^p$.
\end{quote}
\vspace{0.5ex}

Let us write $f_p(n)=a_{p+1} n^{p+1}+a_n p^n+\cdots+a_1 n+a_0$; it can
be proved that $a_0=0$ for all~$p\in\nat$. Clearly finding the
Faulhaber formula, given any $p\in\nat$, can be reduced to finding the
corresponding coefficients:~$a_1,\cdots,a_{p+1}$.

A well-known relation between Faulhaber formulae and Bernoulli
numbers, that is $f_p(n)=\frac{1}{p+1}\sum_{i=0}^p
\binom{p+1}{i}\mathsf{b}_i n^{p+1-i}$ with $\mathsf{b}_i$ being the
$i^\mathit{th}$ Bernoulli number (when~$\mathsf{b}_1=+\frac{1}{2}$),
can be used for computing $a_1,\cdots,a_{p+1}$. This approach however
requires computing Bernoulli numbers $\mathsf{b}_0,
\cdots,\mathsf{b}_p$.

There are various algorithms in the literature for computing Bernoulli
numbers.  These algorithms are generally based on recurrence
relations, where $\mathsf{b}_i$ is computed using $\mathsf{b}_0,
\cdots,\mathsf{b}_{i-1}$, e.g.\ see~\cite{kb,ad}.  In this paper, we
give a recurrence relation on Faulhaber formulae, which yields a
direct algorithm for computing the coefficients $a_1,\cdots,a_{p+1}$.

\paragraph{Structure of paper.}
In section~\ref{sec:algorithm} we give a recursive algorithm for
computing coefficients of Faulhaber formulae. Correctness of the
algorithm is proved in section~\ref{sec:proof}. Time complexity of the
algorithm is also analysed in section~\ref{sec:proof}.

\section{Direct algorithm} 
\label{sec:algorithm}
We are interested in computing the coefficients of the Faulhaber
formula that describes $f_p(n)$, for a given $p\in\nat$.  Write
$f_p(n)=a_{p+1} n^{p+1}+\cdots+a_1 n$.
Let us consider a table in which rows refer to different values
of~$p$, and columns refer to the powers of~$n$. The element at the
intersection of row $i$ and column $j$, denoted $a_{(i,j)}$, is meant
to represent the coefficient of $n^j$ in the polynomial describing
$f_i(n)$. See figure~\ref{fig:int}. Note that elements at $(i,j)$ with
$j>i+1$ are all zero.

\begin{figure}[h]
\begin{center}
\begin{tabular}{|l||l|l|l|l|}
\hline
$i\downarrow\quad j\to$&{\bf 1} & {\bf 2} & {\bf 3} & {\bf 4}\\
\hline
\hline
{\bf 0}&1 & - & - & -\\
\hline
{\bf 1}&$\frac{1}{2}$ & $\frac{1}{2}$ & - & -\\
\hline
{\bf 2}&$\frac{1}{6}$ & $\frac{1}{2}$ & $\frac{1}{3}$&-\\
\hline
{\bf 3}&0 & $\frac{1}{4}$ & $\frac{1}{2}$ & $\frac{1}{4}$\\
\hline
\end{tabular}
\end{center}
\caption{Intuitive description of the direct algorithm}
\label{fig:int}
\end{figure}

Our goal is therefore to find the numbers in the row corresponding
to~$p$. We proceed inductively: First, row~$\mathbf{0}$ is filled,
then we fill row~$\mathbf{1}$, \ldots, till the row numbered
with~$p$ is filled.  The algorithm starts with placing $1$ at
position~$(0,1)$ in the matrix, thus filling row~$\mathbf{0}$. This step
reflects $f_0(n)=1^0+\cdots+n^0=n$.
In order to fill the $i^\mathit{th}$ row, with $i>0$, we follow the
rules below:
\begin{enumerate}
\item
For $1<j\le i+1$, let $a_{(i,j)}=\frac{i}{j} a_{(i-1,j-1)}$.
\item Next, we compute $a_{(i,1)}$ as
$a_{(i,1)}=1-\sum_{1<j\le i+1}a_{(i,j)}$. Put differently, $a_{(i,1)}$
is chosen such that the sum of the numbers that appear in each row
equals~1.
\end{enumerate}
The procedure stops when the row corresponding to~$p$ is
filled. Below, it is proved that $a_{(p,j)}$, for $1\le j\le p+1$, is
the coefficient of $n^j$ in the polynomial that represents $f_p(n)$.
Algorithm~\ref{alg-co} implements this procedure.

Note that, since filling row $i$ only requires elements of row $i-1$,
the algorithm only stores a vector, instead of the matrix of
figure~\ref{fig:int}.

\begin{algorithm}
\caption{Computes coefficients of Faulhaber formulae}
\begin{algorithmic}
\label{alg-co}
\REQUIRE $p\in\nat$
\STATE $a_{1}:=1$
\FOR{$(i:=1;\ i\le p;\ i{++})$}
\STATE s:=0
\FOR{$(j:=i+1;\ j>1;\ j{--})$}
\STATE $a_j:=\frac{i}{j} a_{j-1}$
\STATE $s:= s + a_j$
\ENDFOR
\STATE $a_1:=1-s$
\ENDFOR
\RETURN {$a_1,\cdots,a_{p+1}$}
\end{algorithmic}
\end{algorithm}

Using the table of figure~\ref{fig:int} and the presented algorithm, 
we make the following simple observations about coefficients of
Faulhaber formulae.
\begin{itemize}
\item The coefficient of $n^{p+1}$ in $f_p(n)$ is $\frac{1}{p+1}$, for
any $p\in\nat$. This can be proved by induction:
$a_{(p,p+1)}=\frac{p}{p+1}a_{(p-1,p)}$ and $a_{(1,2)}=\frac{1}{2}$.
\item The coefficient of $n^p$ in $f_p(n)$ is $\frac{1}{2}$, for any
$p\in\nat$. This can be proved by induction:
$a_{(p,p)}=\frac{p}{p}a_{(p-1,p-1)}$ and $a_{(1,1)}=\frac{1}{2}$.
\item The coefficient of $n^{p-2}$ in $f_p(n)$ is zero, for any $p\ge
3$. This can be proved by induction:
$a_{(p,p-2)}=\frac{p}{p-2}a_{(p-1,p-3)}$ and $a_{(3,1)}=0$.
\end{itemize}

\section{Recurrence relation on Faulhaber formulae}
\label{sec:proof}
In this section we prove that algorithm~\ref{alg-co} correctly
computes the coefficients of Faulhaber formulae. For this, first, a
recurrence relation on Faulhaber formulae is proved.

\begin{lemma}[Recurrence on Faulhaber formulae]
\label{main-thm}
\[
\begin{array}{l}
\bullet\ f_0(n)=n,\\ 
\bullet\ \mbox{for}\ p>0,\
f_{p}(n)=p\int^n_0 f_{p-1}(t)\ dt +
(1-p\int^1_0 
f_{p-1}(t)\ dt)n\ 
\end{array}
\]
\end{lemma}
\begin{proof}
The first part, $f_0(n)=1^0+\cdots+2^0=n$, can be proved by
straightforward induction. In the following, therefore, we assume
$p>0$. To prove the second part, we recall the following relations
(e.g.\ see~\cite[chapter 23]{mhandbook}):
\begin{enumerate}[I]
\item $f_{p-1}(n)=\frac{1}{p}(B_{p}(n+1)-B_{p}(0))$, for $p\in\nat^+$.
\item 
$\int_a^b B_i(t) dt = \frac{1}{i+1}(B_{i+1}(b)-B_{i+1}(a))$.
\item
$B_i(n+1)-B_i(n)=i n^{i-1}$, for $n\in\nat, i>1$.
\end{enumerate}
where $B_i(t)$ is the $i^\mathit{th}$ Bernoulli polynomial. 
Using (I) and (II), we get:
\[
\begin{array}{l}
\int_a^b p f_{p-1}(t)\ dt\\
=\int_a^b B_{p}(t+1)-B_{p}(0)\ dt\\
=\frac{1}{p+1}(B_{p+1}(b+1)-B_{p+1}(a+1))- B_p(0) (b-a)
\end{array}
\]
Therefore
\[
\begin{array}{l}
1-\int^1_0 p f_{p-1}(t)\ dt\\
=1-\frac{1}{p+1}(B_{p+1}(2)-B_{p+1}(1))+ B_p(0)\\
=B_p(0)
\end{array}
\]
The last simplification step is due to~(III). Note that since $p>0$, we
have $p+1>1$, satisfying the precondition of (III).
Similarly
\[
\begin{array}{l}
\int^n_0 p f_{p-1}(t)\ dt\\
=\frac{1}{p+1}(B_{p+1}(n+1)-B_{p+1}(1))-n B_p(0)
\end{array}
\]
As a result
\[
\begin{array}{l}
\int^n_0 p f_{p-1}(t)\ dt + (1-\int^1_0 p
f_{p-1}(t)\ dt)n\\
=\frac{1}{p+1}(B_{p+1}(n+1)-B_{p+1}(1))\\
=\frac{1}{p+1}(B_{p+1}(n+1)-B_{p+1}(0))
\end{array}
\]
The last simplification step is again due to~(III):
$B_{p+1}(1)-B_{p+1}(0)=0$, with $p+1>1$.
Finally
\[
\begin{array}{l}
f_p(n) \\
= \frac{1}{p+1}(B_{p+1}(n+1)-B_{p+1}(0))\\
=\int^n_0 p f_{p-1}(t)\ dt + (1-\int^1_0 p
f_{p-1}(t)\ dt)n
\end{array}
\]
This completes the proof.
\end{proof}

Now, we are ready to prove the correctness of algorithm~\ref{alg-co}.
\begin{theorem}[Correctness]
Given $p\in\nat$, algorithm~\ref{alg-co} outputs the
coefficients of the Faulhaber formula that represents $f_p(n)$.
\end{theorem}
\begin{proof}
Let us assume the coefficient of $n^j$ is $\alpha$ in $f_i(n)$, for
some~$1<j\le i+1$, and the coefficient of $n^{j-1}$ in $f_{i-1}(n)$
is~$\beta$. From the recurrence relation of lemma~\ref{main-thm}, we
get $\alpha=\frac{i}{j}\beta$. This is simply because $\int_0^n t^k\
dt= \frac{1}{k+1} n^{k+1}$, where $k$ is a any positive rational
number. This directly results in the way  algorithm~\ref{alg-co}
recursively computes $a_{(i,j)}$, for $1<j\le i+1$, and $i>0$.

Now, note that $f_p(1)=1$, for any $p\in\nat$. Moreover, note that
$f_p(1)=a_{p+1}+\cdots+a_1$. Therefore,
$a_1=1-(a_{p+1}+\cdots+a_2)$. This immediately results in the way
algorithm~\ref{alg-co} computes $a_{(i,1)}$, for $i>0$.
\end{proof}

Below, we turn to time complexity of algorithm~\ref{alg-co}.  To
measure the computational complexity, we count the number of
multiplication and addition (or, subtraction) operations that are
performed on rational numbers.  Assignments to constants, and
incrementing and decrementing natural numbers (i.e.\ counters in the
algorithm) are thus assumed to take negligible time.

\begin{theorem}[Time complexity]
Time complexity of the direct algorithm is quadratic in $p$.
\end{theorem}
\begin{proof}
Note that the outer \textbf{for} loop is repeated $p$ times, and the
inner \textbf{for} loop is repeated $1+\cdots+p=\frac{1}{2} p (p+1)$
times. It is straightforward to see that, for $p\in\nat$, the number
of addition operations is $\frac{1}{2}p(p+1)+p$, and the number of
multiplication operations that are performed on rational numbers is
$\frac{1}{2}p(p+1)$.
\end{proof}

\bibliographystyle{alpha}
\bibliography{bib}

\begin{thebibliography}{AD09}

\bibitem[AD09]{ad}
T.~Agoh and K.~Dilcher.
\newblock Shortened recurrence relations for {B}ernoulli numbers.
\newblock {\em Discrete Mathematics}, 209(4):887--898, 2009.

\bibitem[AS72]{mhandbook}
Milton Abramowitz and Irene Stegun, editors.
\newblock {\em Handbook of Mathematical Functions with Formulas, Graphs, and
  Mathematical Tables}.
\newblock Dover, New York, 1972.

\bibitem[KB67]{kb}
Donald Knuth and Thomas Buckholtz.
\newblock Computation of tangent, {E}uler, and {B}ernoulli numbers.
\newblock {\em Mathematics of Computation}, 21(100):663--688, 1967.

\end{thebibliography}

\end{document}